\documentclass[journal,12pt,onecolumn,draftclsnofoot,]{IEEEtran}
\usepackage[T1]{fontenc} 
\usepackage[english]{babel} 
\usepackage{amsmath,amsfonts,amsthm,amssymb} 
\usepackage{algorithm}
\usepackage{enumitem}
\usepackage[noend]{algpseudocode}
\usepackage{cite,graphicx,float}
\usepackage{subcaption}
\usepackage{color}
\usepackage{caption}
\usepackage{soul}
\usepackage[figurewithin=none]{caption}

\usepackage{fancyhdr} 
\newtheorem{theorem}{Theorem}
\newtheorem{lemma}{Lemma}

\title{	\Large Energy-efficient Wireless Charging and Computation Offloading In MEC Systems\\ 
}

\author{Rafia Malik and Mai Vu\\
\small Department of Electrical and Computer Engineering, Tufts University, MA, USA\\
Email: rafia.malik@tufts.edu, mai.vu@tufts.edu} 

\date{\normalsize{May 21, 2018}} 

\begin{document}
\maketitle 

\begin{abstract}
Wireless charging coupled with computation offloading in edge networks offers a promising solution for realizing power-hungry and computation intensive applications on user devices. We consider a mutil-access edge computing (MEC) system with collocated MEC servers and base-stations/access points (BS/AP) supporting multiple users requesting data computation and wireless charging. We propose an integrated solution with computation offloading to satisfy the largest proportion of requested wireless charging while keeping the energy consumption at the minimum subject to the MEC-AP transmit power and latency constraints. We propose a novel algorithm to perform data partitioning, time allocation, transmit power control and design the optimal energy beamforming for wireless charging. Our resource allocation scheme offers an energy minimizing solution compared to other schemes while also delivering higher amount of transferred charge to the users.
\end{abstract}

\section{Introduction}
In recent years, there has been a significant rise in the number of connected devices, coupled with a rampant growth of wireless networks. The large number of connected devices has led to an evolution of wireless communication networks towards dense deployments with an exponential growth in wireless traffic. The increased data traffic and the use of applications requiring high data rate on cell phones has led to an escalated energy demand. Future generation networks including 5G and beyond are expected to handle this multiple folds increase of data traffic at stringent latency requirements with a need for energy conservation~\cite{Krikidis2014}. Therefore, in addition to performance parameters like throughput, coverage and latency, energy efficiency can be considered as a figure of merit in the design of next generation wireless networks.

\begin{figure}[t]
\centering
\includegraphics[scale = 0.6]{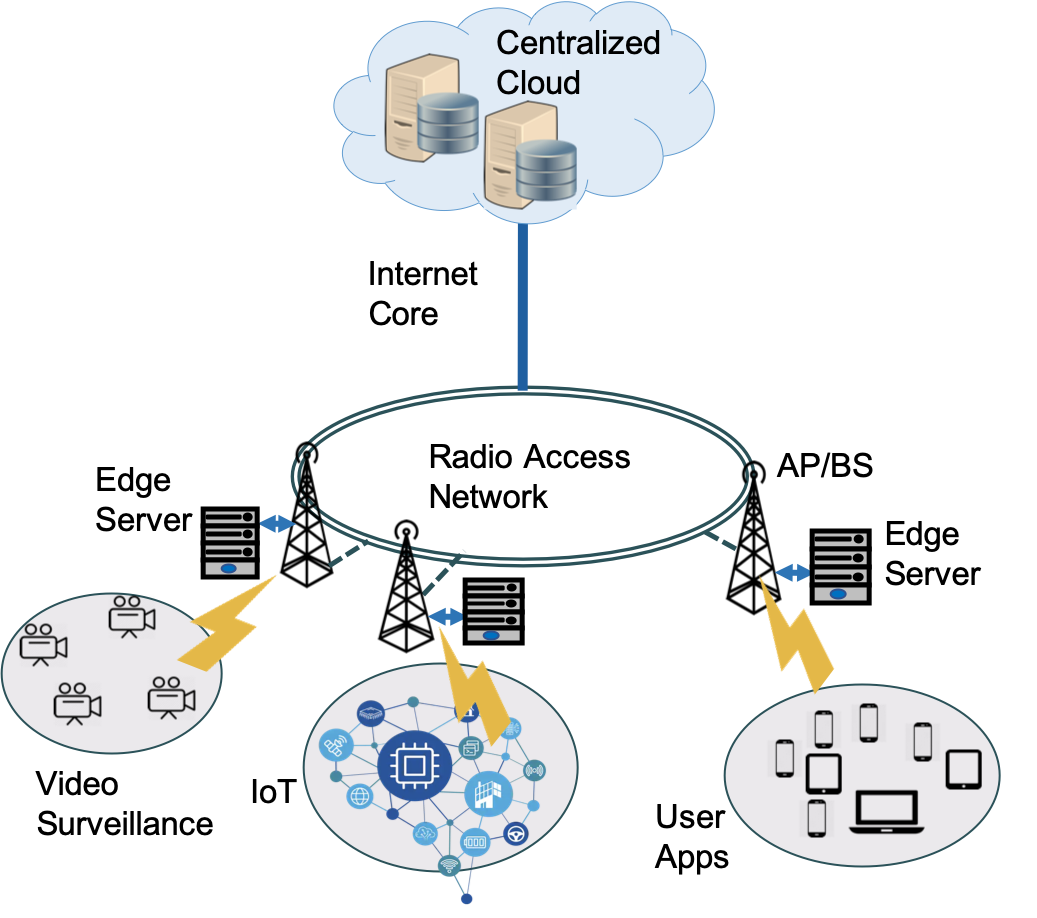}
\caption{Architecture for an edge computing system providing computation offloading and wireless charging to connected users}
\label{mec_arch}
\end{figure}
Regardless of expected improvements in Quality of Service (QoS) and resilience in future networks like 5G, huge volumes of video traffic will continue to present significant challenges for mobile network operators~\cite{Nightingale2018}. To reach the desired improvement over existing cellular networks, and to implement key use-cases featuring high data-rates and low latency, Multi-access Edge Computing (MEC) is a promising technology which can provide cloud-computing capabilities within the radio access network in close vicinity to mobile subscribers~\cite{MEC2014}. Figure \ref{mec_arch} shows a general architecture for an edge computing network. There can be several scenarios under which computation offloading occurs, for instance video surveillance cameras offloading to the edge, or IoT devices or applications like AR/VR offloading their computation intensive tasks to the MEC servers. These servers can be co-located with radio base stations connected via backhaul to the internet core which is connected to the centralized cloud \cite{MEC2018}. By moving the computing features to the edge, MEC can offer a distributed and decentralized service environment characterized by proximity, low latency, and high rate access~\cite{Doppler2018}\cite{Cau2016}. Currently, ETSI industry specification group is the only international standard available for MEC in the technology field, however, the 3rd Generation Partnership Project (3GPP) has started to include MEC in the 5G network standardization \cite{Intel2019}.

Radio frequency (RF) energy harvesting has lately garnered significant interest for communication systems with the prospect of far-field wireless power transfer which can enable energy-constrained devices to replenish their charge levels without physical connections. Energy harvesting is projected to be the next billion dollar market for semiconductors~\cite{Semico2016}. Progress in this area includes global coalition initiatives such as Airfuel Alliance for building a global, interoperable wireless power ecosystem \cite{Airfuel} as well as commercial products, for example, Ossia's Cota technology which uses RF power to charge dozens of mobile devices within a several meter radius and their transmitters come in multiple form factors for instance as a ceiling tile~\cite{Ossia}. There is also active research in RF power transfer ranging from signal design to maximize energy harvesting potential \cite{Bruno2019} to application centric research for using UAVs for wireless charging \cite{Rui2019}.

The availability of Ultra-High-Definition portable consumer devices and AR/VR applications fuels the growth of mobile video traffic, however, the limited battery lifetime of these devices poses a hindrance to the deployment of such power-hungry designs and computation intensive features~\cite{Barbar2015}. To this end, the synergy between edge computing and wireless power transfer has the potential to provide battery sustainability and to alleviate the computation load. Dense deployments of multiple base-stations with co-located MEC servers \cite{MEC2018} in close proximity to connected users can warrant the practicality of wireless charging and offer high access rates and computation capabilities. Prior works have considered wireless charging in MEC systems under different implementations, for instance, wireless charging in cooperation assisted edge computing~\cite{Hu2018}, UAV-enabled mobile edge computing~\cite{Chu2018} and MEC based heterogeneous networks~\cite{Ji2018}. Wireless power transfer has been considered in MEC networks for \textit{self-sustained} devices, which rely on wireless charging as their sole power source, in relay-aided edge systems~\cite{Hu2018}, single user~\cite{Chae2016} and multiple user systems~\cite{Wang2018}. Different from the concept of self-sustained devices which typically have low power requirements and/or low receiver sensitivity, \textit{on-request} wireless charging can be more widely applicable where user devices use wireless charging to replenish their batteries. In the case of cellular networks, providing charging may can be a billable service assuming that the ground users have knowledge of their battery state, and can inform the MEC-AP about their battery level for requesting recharge in cases where their battery is critically low.

Computation offloading to the edge has been studied under two data models;, binary offloading where the task is completely offloaded to the MEC for computation, or kept entirely at the user end for local computation, and partial offloading where the task can be disintegrated such that some of it is offloaded to the MEC and the remaining is computed locally. Modern mobile applications are composed of numerous procedures, for example, an AR/VR application can have multiple computation components such as video rendering, mapping and tracking, object recognition, etc, which makes it possible to implement fine-grained (partial) computation offloading. While partial offloading is more complicated to implement compared to the traditional binary offloading scheme, however, it is more realistic with possible implementation in practical edge computing systems and has immense benefits in terms of energy consumption as shown in \cite{Malik2020}.

In this work, we consider a multi-cell multi-user network scenario where access points equipped with massive MIMO antenna arrays and with co-located mobile edge computing servers offer computation offloading and wireless charging. We consider integrated computation offloading and wireless charging with the aim to minimize the transmitted energy consumption while ensuring that the received energy is the largest feasible proportion of the requested energy. This is different from our previous work in [Malik-maxcharging] which proposes an opportunistic wireless charging scheme designed to maximize the received (charged) energy at the user end under the latency and transmit power constraints, but does not aim at minimizing the energy consumption for wireless charging. The opportunistic wireless charging in [Malik-maxcharging] is achieved by formulating a sequential problem where the computation offloading resources are optimized independent of the wireless charging resources. Results show that while opportunistic schemes can yield higher wireless charging capability, it is at the cost of higher charging energy. In this paper, we formulate an integrated problem for joint resource allocation of data computation, communication and wireless charging resources with the aim of minimizing the overall system's energy consumption.

\subsection*{Major Contributions}
\begin{enumerate}[leftmargin=*]
\item We formulate a novel comprehensive system-level energy minimization problem in a multi-cell scenario where each cell has an MEC server which provides wireless charging and computation offloading services simultaneously for multiple users. While previous works have considered MEC problems, most consider energy consumption only at the AP~\cite{Wang2018,Chu2018} or the user end~\cite{Scutari2015,Chae2016,You2017}. On the other hand, our formulation minimizes a weighted sum of the energy consumption at all users in each cell and its MEC server and therefore includes the energy optimization at either the users or the MEC side alone as its special cases. This is the first work to design an integrated system-level problem for joint resource allocation of data transmission, computation and wireless charging resources with constraints on latency and power.

\item Another key novelty of our work is the all-inclusive set of optimizing variables considered in the problem. Our formulation optimizes for the time durations for offloading, downloading and computation, transmit power allocation at both users and APs, and the optimal split of data to be offloaded and to be computed locally at each user. Most prior works consider binary offloading~\cite{Guo2018,Leng2016,Letaief2017,Sengul2017} and no optimization for downloading and/or computation time~\cite{Chae2016,You2017,Bi2018,Wang2018}.

\item We design a novel solution approach to solve the complex non-convex integrated energy optimization problem and propose a nested algorithm architecture for optimal resource allocation for both computation offloading as well as wireless charging. We disintegrate the problem into equivalent subproblems and propose a customized algorithm architecture and implementation which is unique. It includes an outer latency-aware descent algorithm which solves for optimal data partitioning, and inner primal-dual algorithm which jointly optimizes the time allocation and energy beamforming matrices. 
\end{enumerate}

\subsubsection*{Notation} $\boldsymbol{X}$ and $\boldsymbol{x}$ denote a matrix and vector respectively, $\nabla^2 f(x)$ denotes the Hessian matrix, and $\nabla^2 f(x)^{-1}$ denotes its inverse. For an arbitrary size matrix, $\boldsymbol{Y}$,  $\boldsymbol{Y}^\ast$ denotes the Hermitian transpose, and $\textbf{diag}(y_1,...,y_N)$ denotes an $N\times N$ diagonal matrix with diagonal elements $y_1,...,y_N$. $\boldsymbol{I}$ denotes an identity matrix, and $\boldsymbol{0, 1}$ denote an all zeros and all ones vector respectively. The standard circularly symmetric complex Gaussian distribution is denoted by $\mathcal{CN}(\boldsymbol{0}, \boldsymbol{I})$, with mean $\boldsymbol{0}$ and covariance matrix $\boldsymbol{I}$. $\mathbb{C}^{k \times l}$ and $\mathbb{R}^{k \times l}$ denote the space of $k \times l$ matrices with complex and real entries, respectively.

\section{System Model}\label{sys_model}
We consider a system where $L \geq 1$ Access Points (APs), each co-located with an MEC Server, are deployed over a targeted zone/area, for instance in a sports stadium or a town fair, serving ground users with computation offloading and power transfer. Each AP is equipped with a massive antenna array with $N$ antennas while the user-devices are equipped with single antennas. These APs wirelessly charge (upon request) ground users in downlink, collect offloaded data from the users in uplink, and deliver computed results to users in downlink~\cite{MEC2014}. We consider $K$ users requesting wireless charging service and sending data for computation offloading to each MEC-AP. 

For computation offloading at each MEC, we consider the simple \textit{data-partition model}, where the task-input bits are bit-wise independent and can therefore be arbitrarily divided into different groups to be executed by different entities~\cite{Mao2017}. We consider the case of partial offloading, such that for the $ith$ user, the $u_i$ computation bits are partitioned into $q_i$ and $s_i$ bits, where $q_i$ bits are computed locally and $s_i$ bits are offloaded to the MEC server. Assuming that such partition at the user-terminal does not incur additional computation bits, then $u_i = q_i + s_i$. 

\begin{figure}[t]
\centering
\includegraphics[scale = 0.6]{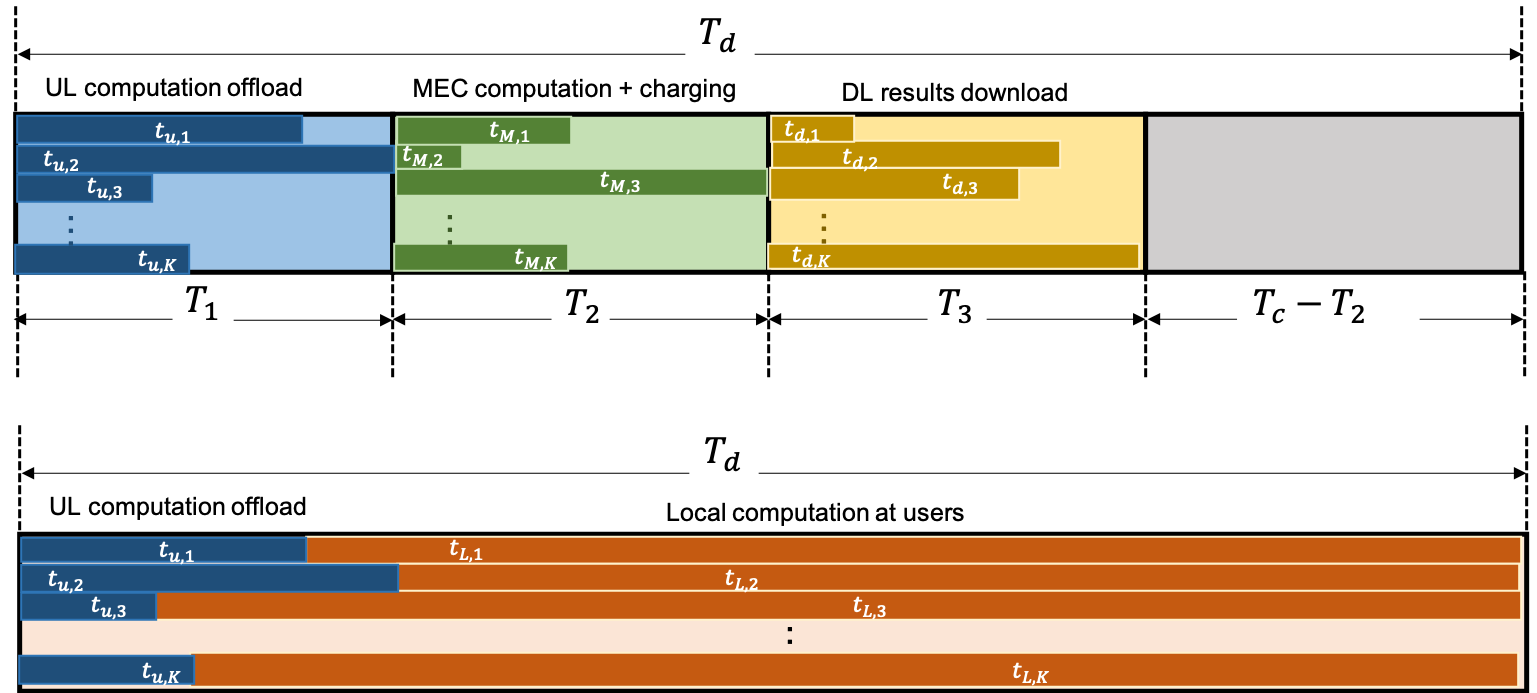}
\caption{System functions at the MEC and user end that take place within the latency constraint $T_d$}
\label{Td_DL_UL}
\end{figure}
\textcolor{blue}{Energy at the user terminal is consumed for two tasks; 1) for local computation which depends on the CPU frequency used, and 2) for transmitting the data for computation offloading to the serving MEC-AP in the uplink which depends on the transmission time and power. Energy at the MEC server is consumed for three tasks; 1) for data computation of offloaded tasks, 2) for transmitting the results of computed data to its users in the downlink, and 3) for wireless charging in the downlink to the users requesting energy. Consider the case where wireless charging is requested jointly with computation offloading. Given a latency constraint of $T_d$, the time span for data offloading, computation at both the users and the MEC ends, wireless charging, and delivery of computed results to the user should not exceed $T_d$. From the MEC-AP's perspective, the time duration for data offloading from all users to the MEC is denoted by $T_1$, the time for wireless charging is denoted as $T_c$, the computation for offloaded data at the MEC spans duration $T_2$, and the transmission of processed results occupies time $T_3$, such that $\sum_{j=1}^3 T_j \leq T_d$. Note that wireless charging can happen concurrently with data computation at the MEC, and can continue after the results have been delivered to the users in downlink. From the user's side, the time taken for data offloading by the $i^{\text{th}}$ user, $t_{u,i}$, and the time taken for local computation of any remaining data for this user, $t_{L,i}$ should meet the latency constraint, such that $t_{u,i}+t_{L,i} \leq T_d$. Figure \ref{Td_DL_UL} shows these system functions from both the users' and the MEC's perspective.}

\subsection{Wireless Charging}
In each cell, we consider $K$ user terminals requesting wireless charging from the MEC-AP, where the $i^{\text{th}}$ user requests $e_i$ mJ of energy. To cater for the energy requests from the multiple users, the massive-MIMO enabled MEC-AP employs transmit energy beamforming in downlink to simultaneously charge multiple users as shown in Figure \ref{energybeam}. 
\begin{figure}[t]
\centering
\includegraphics[scale = 0.6]{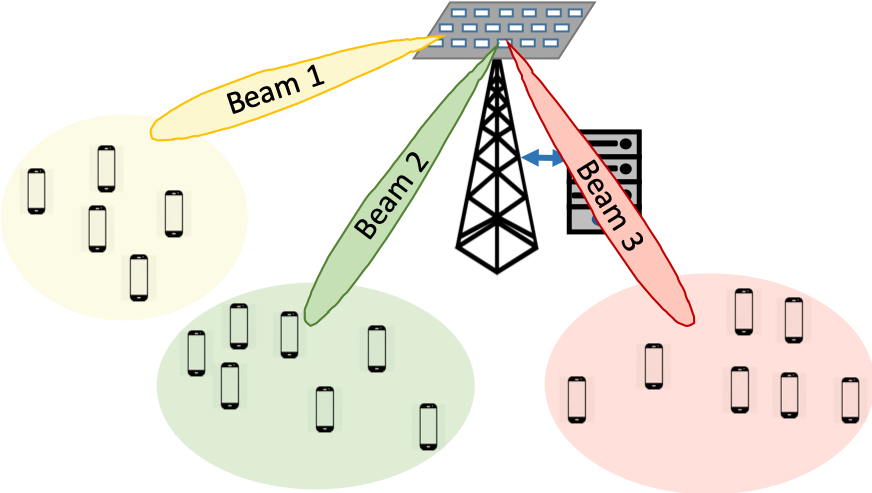}
\caption{MEC-AP wirelessly charging multiple users simultaneously through energy beamforming}
\label{energybeam}
\end{figure}

Let $\boldsymbol{x_q}$ denote the energy bearing signal from the AP to the user-terminal (UT), $\boldsymbol{W_q} \triangleq \mathbb {E}\Big[\|\boldsymbol{x_q}\|^2\Big]$ denote the transmit covariance matrix, and $P_c = \text{tr}(\boldsymbol{W_q})$ be the power transmitted from the AP for wireless charging, in short, the charging power. Then the received (charged) power at the $i^{\text{th}}$ user is given as
\begin{equation}
P_{h,i} = \xi \mathbb{E} \left [ \left | \boldsymbol{h_i^\ast x_q}\right |^2\right ] = \xi \text{tr} (\boldsymbol{h_i^\ast W_q h_i})
\end{equation}
where $0 \leq \xi \leq 1$ is the energy conversion efficiency from Radio Frequency (RF) to Direct Current (DC), $\boldsymbol{h_i} \in \mathbb{C}^{N \times 1}$ is the channel from the AP to the $i^{\text{th}}$ user.  We define $T_c$ as the time duration for wireless charging, where $T_c = T - (T_1 + T_3)$ and includes the time consumed by the computation phase, over which power is transferred to the users alongside computation at the users and at the MEC server. The energy consumed at the MEC server for power transfer, in short the charging energy, is given by
\begin{equation}\label{Ec}
E_c = T_c \text{tr}(\boldsymbol{W_q})
\end{equation}

For the $i^{\text{th}}$ user requesting $e_i$ amount of energy, the received (charged) energy, $E_{h,i}$, is given as
\begin{equation}\label{cond_energy}
E_{h,i} = P_{h,i} T_c = \xi T_c \text{tr} (\boldsymbol{h_i^\ast W_q h_i}) \geq \alpha_i e_i
\end{equation}
Here $0 \leq \alpha_i \leq 1 \ \forall i\in[1,K]$ is defined as an \textit{energy ratio} auxiliary variable to ensure that the received energy is proportional to the requested amount such that only a portion of the requested energy may be charged if it is unfeasible for the AP to satisfy the user's energy request completely due to poor channel conditions or high energy request(s) by a single or few users. The ratio variable $\alpha_i$ therefore serves several reasons: to avoid over charging users' batteries, to conserve energy and spending since charging is a billable service, and also to ensure fairness among users so that no single user gets an unfairly large amount of the charged energy at the expense of others. 
\subsection{Data Transmission}
\subsubsection{Offloading Data in Uplink}
In a given time slot, $K$ single-antenna user terminals simultaneously offload to the $N$ antenna AP. We consider $N \gg K$ such that the throughput becomes independent of the small-scale fading with channel hardening~\cite{Ngo2017}. The very large signal vector dimension at a massive MIMO AP enables the use of linear detectors such as maximum ratio combining (MRC), in which case the uplink net achievable transmission rate for the $i^{th}$ user in the $l^{th}$ cell, $r_{u,i}$, is given as~\cite{Marzetta2016}
\begin{equation}\label{rate_ul}
r_{u,i} = \nu \log_2 \left ( 1 + \frac{\text{SINR}_{li}^{ul}}{\Gamma_{1}} \right ), \ \text{SINR}_{li}^{ul} = \frac{N \gamma_{li}^l p_{li}}{\sigma_{1,li}^2}
\end{equation}
where $\Gamma_{1} \geq 1$ accounts for the capacity gap due to practical coding schemes, $\gamma_{li}$ is the mean-square channel estimate, and $p_{li}$ is the transmit power of the $i^{\text{th}}$ user in the $l^{\text{th}}$ cell.  The constant $\nu$ represents the portion of transmission symbols spent on data transfer in the coherence interval $\tau_c$. The interference and noise power, $\sigma_{1,li}^2$, includes the receiver noise variance, interference due to channel estimation and from contaminating cells, and inter-cell interference as defined in~\cite[Eq. 4.18]{Marzetta2016}.

The energy consumed for offloading the $i^{th}$ user's data is given by $E_{OFF,i} = p_i t_{u,i}$, where $p_i$ is the transmit power and $t_{u,i}$ is the transmission time for the $i^{th}$ user. Let $B$ denote the channel bandwidth, then $t_{u,i} = \frac{s_i}{B r_{u,i}}$. All users offload their computation bits simultaneously, and the total energy and time overhead for simultaneous data offloading is given as
\begin{equation}\label{E_ul}
E_{OFF} = \sum_{i=1}^{K} \frac{p_i s_i} {B r_{u,i}}, \ T_1 = \max_{i \in [1,K]} t_{u,i}.
\end{equation}

\subsubsection{Downloading Results in Downlink}
For the $i^{th}$ user in the $l^{th}$ cell, the downlink transmission rate with maximum ratio linear precoding at the MEC-AP is given as~\cite{Marzetta2016}
\begin{equation}\label{rate_dl}
r_{d,i} = \log_2 \left ( 1 + \frac{\text{SINR}_{li}^{dl}}{\Gamma_{2}} \right ), \ \text{SINR}_{li}^{dl} = \frac{N P \gamma_{li}^l \eta_{lk}}{\sigma_{2,li}^2}
\end{equation}
where $\Gamma_{2} \geq 1$ is the capacity gap, and $\sigma_{2,li}^2$ is the interference and noise power which also contains pilot contamination and intercell interference as given in~\cite[Eq. 4.34]{Marzetta2016}.

The transmission time for delivering the $i^{th}$ user's computation results can be written in terms of the downlink rate in (\ref{rate_dl}) as $t_{d,i} = \frac{\tilde{s}_i}{B r_{d,i}}$. Here $\tilde{s}_i$ denotes the number of information bits generated after processing $s_i$ offloaded bits of the $i^{th}$ user, and is assumed to be proportional to $s_i$, that is $\tilde{s}_i = \mu s_i$.  The AP simultaneously transmits computed results for all users, and the total energy and time overhead for results downloading are then given as
\begin{equation}\label{E_dl}
E_{DL} = \sum_{i=1}^{K} \frac{P \eta_i \mu s_i}{B r_{d,i}}, \ T_3 = \max_{i \in [1,K]} t_{d,i}.
\end{equation}

\subsection{Data Computation}
\subsubsection{Local computation at the users}
The time for computation depends on the amount of data to be computed and the CPU cycle frequency. The energy consumption and the processing time for local computation at the $i^{th}$ user is given as~\cite{Mao2017} 
\begin{align}\label{t_Li}
&E_{LC} = \sum_{i=1}^{K} \kappa_i c_i (u_i - s_i) f_{u,i}^2, \ \ t_{L,i} = \frac{c_i (u_i - s_i)}{f_{u,i}}
\end{align}
where $\kappa_i$ is the effective switched capacitance, $f_{u,i}$ denotes the average CPU frequency, $c_i$ denotes the CPU cycle information, and $q_i = u_i - s_i$ is the total number of bits required to be locally computed at $i^{th}$ user respectively. The users' local computation time can also extend to Phase III while the MEC is sending computed results back to users. This fact is considered later in the problem formulations.

\subsubsection{Computation of the offloaded data at the MEC server}
MEC servers, with high computation capacities, compute the tasks of all users in parallel~\cite{Taleb2017}\cite{Mao2017}. The energy and time consumed for computing offloaded bits is given as
\begin{equation}\label{tMEC}
E_{OC} = \sum_{i = 1}^{K} \kappa_m f_{mi}^2 d_m s_i, \ \ t_{M,i} = \frac{d_m s_i}{f_{mi}} \ \forall i \in [1, K], \ \ T_2 = \max\{t_{M,i}\}.
\end{equation}
where $t_{M,i}$ is the time for computing $i^{th}$ user's offloaded task, $s_i$ is the number of bits offloaded by the $i^{th}$ user to the MEC, $d_m$ is the number of CPU cycles required to compute one bit at the MEC, $f_{mi}$ is the CPU frequency assigned to the $i^{th}$ user's task, and $\kappa_m$ is the effective switched capacitance of the MEC server.

For our formulation to follow in Section \ref{formulation}, we consider equal frequency allocation for users' tasks, that is $f_{m,i} = f_m \ \forall i$, based on previous results in [ref-MEC] showing that dynamic frequency allocation has little effect on the system energy consumption since in a typical network setting, the wireless transmission energy consumption is significantly dominant compared to the computation energy consumption.

\section{Energy Minimization Problem}\label{formulation}
Considering a multi-cell multi-MEC network, we formulate an edge computing problem which explicitly accounts for physical layer parameters including available transmit powers from each user and the MEC, associated massive MIMO data rates with realistic pilot contamination and interference. For simplicity of notation, we assume that all $K$ users which are offloading their computation to the MEC server are also requesting wireless charging. 

\textcolor{blue}{The total energy consumption by all users, based on equations (\ref{E_ul}) and (\ref{t_Li}), can be written as
\begin{equation}\label{E_u}
E_u =  \sum_{i=1}^{K} \left[\frac{t_{u,i}(2^{\frac{s_i}{\nu t_{u,i}B}} - 1)\Gamma_{1}\sigma_{1,i}^2}{N \gamma_i} + \kappa_i c_i (u_i - s_i) f_{u,i}^2 \right]
\end{equation}
Similarly, the total energy consumption at the MEC server, based on equations (\ref{Ec}), (\ref{E_dl}) and (\ref{tMEC}) is
\begin{equation}\label{E_m}
E_m = \sum_{i=1}^{K} \left[\frac{t_{d,i}(2^{\frac{\mu s_i}{t_{d,i}B}} - 1)\Gamma_{2}\sigma_{2,i}^2}{N \gamma_{i}}  + \kappa_m d_m f_{mi}^2 s_i \right] + (T_d - T_1 - T_3) \text{tr}(\boldsymbol{W_q})
\end{equation}
In these expressions, using (\ref{rate_ul}) and (\ref{rate_dl}), and by definition of the uplink and downlink transmission rates as $r_{u,i} = \frac{s_i}{\nu t_{u,i} B}$ and $r_{d,i} = \frac{\mu s_i}{t_{d,i} B}$ respectively, we have implicitly replaced the power allocation variables for per-user uplink transmission power ($p_{li}$) and per-user downlink power ($\eta_{li}$) as functions of the time allocation and data partitioning as follows
\begin{align}\label{poweralloc}
p_{li} = \frac{(2^{\frac{s_i}{\nu t_{u,i} B}} - 1)\Gamma_{1}\sigma_{1,i}^2}{N \gamma_{i}}, \  \ \eta_{li} = \frac{(2^{ \frac{\mu s_i}{t_{d,i} B}} - 1)\Gamma_{2}\sigma_{2,i}^2}{P N \gamma_{i}}
\end{align}} Below we discuss an integrated formulation which jointly optimizes for the wireless charging transmit beamforming matrix, the amount of data offloaded from each user, and the time duration for each phase within a total latency requirement with aim of system level energy minimization. \textcolor{blue}{This is different from our work in [Malik-maxcharging] where we formulate two sequential problems, one with the aim of minimizing the energy consumption for computation offloading, and the other problem to maximize the energy received by the users through wireless charging. So the objective function in this formulation includes the energy consumed for wireless charging as in (\ref{E_m}) with the aim of minimizing the overall energy consumption, however, for the sequential formulation in [Malik-maxCharging], in terms of the wireless charging aspect, the objective is to maximize the total energy received by all users, $\sum_{i=1}^K E_{h,i}$.}

\setcounter{equation}{13}
\begin{align*}\label{Pnew}
(P_\text{int}): \ \  &\min_{\boldsymbol{s, t, W_q}} E_{\text{total}} = (1 - w) E_{u} + w E_{m}  \tag{\theequation}&\\
&\textcolor{blue}{\text{ s.t. }  \ \ \  \text{Eqs. } (\ref{E_u})-(\ref{E_m})} \tag{a-b}&\\
&\sum_{j=1}^{3} \left ( T_j\right ) \leq T_d, \ \ \ \ \frac{c_i (u_i - s_i)}{f_{u,i}} + t_{u,i} - T_d \leq 0 \ \ \ \  \ \ \ \  \forall i \in [1,K] \tag{c-d}&\\
&t_{u,i} - T_1 \leq 0, \ \ \ \  \ t_{d,i} - T_3 \leq 0, \ \ \ \ \frac{d_m s_i}{f_{mi}} - T_2 \leq 0 \ \ \ \forall i \in [1,K] \tag{e-g}&\\
&T_c = T_d - T_1 - T_3 \tag{h}&\\
&\text{tr} (\boldsymbol{W_q}) - P\leq 0, \ \ \ \  \xi \text{tr} (\boldsymbol{h_i^\ast W_q h_i})T_c - \alpha_i e_i  \geq 0 \tag{i-j}&
\end{align*}
Here $E_{\text{total}}$ is weighted sum of energy consumed at all users ($E_{u}$) and the MEC ($E_{m}$), with $1 - w$ and $w$ as the respective weights. The optimizing variables of this problems are time allocation $\boldsymbol{t} = [t_{u,1}...t_{u,K}, t_{d,1}...t_{d,K}, T_1, T_2, T_3, T_c]$, offloaded data $\boldsymbol{s} = [s_1...s_K]$, and beamforming matrix for wireless charging $\boldsymbol{W_q} \in \mathbb{R}^{N \times N}$. Given parameters of the problems are $T_d$ as the total latency constraint, $P$ as the AP's transmit power, $B$ as the channel bandwidth, $\Gamma_1$, $\Gamma_2$ as the uplink and downlink capacity gaps, $(\kappa_i, c_i)$ and $(\kappa_m, d_m)$ as the switched capacitance and CPU cycle information at the users and the MEC respectively.

Here $T_d$ is the total latency constraint, and (c) represents the constraint that both the time consumed for all three phases at the MEC, and the time consumed for offloading $\boldsymbol{t_u}$ and local computation at each user $\boldsymbol{t_L}$ should not exceed $T_d$. Constraints (e-g) show that the time consumed separately for offloading $\boldsymbol{t_u}$, computation of users' tasks at the MEC $\boldsymbol{t_M}$, and downloading time $\boldsymbol{t_d}$ for each user's results must be less than the maximum allowable time, $\{ T_1, T_2, T_3\} $, for that phase as given in \{(\ref{E_ul}),(\ref{tMEC}), (\ref{E_dl})\} respectively. Constraint (h) denotes that wireless charging occupies all the time within $T_d$ outside the data transmission operations. Constraint (i) represents the physical layer constraint on the maximum transmission power of the AP. Constraint (j) shows that the amount of received (charged) energy at the $i^{\text{th}}$ user is proportional to the amount of requested energy depending on the availability of the system, where the proportional factor $\alpha_i \ (0 \leq \alpha_i \leq 1)$ is an auxiliary variable which ensures feasibility for cases when the actual charged energy is less than the requested amount if the available time or MEC-AP power for wireless charging cannot satisfy the full requested energy amount. We are interested in the largest values of $\alpha_i$ for which the problem is feasible, as such these $\alpha_i$ values will be optimized separately. \textcolor{blue}{Note that constraints (h-j) are not included in the energy minimization problem in [Malik-maxCharging] where the beamforming optimization follows the energy minimization for computation offloading, as a separate formulation.}

\subsection{Problem Analysis and Decomposition}
In this section we analyze the integrated problem $(P_{\text{int}})$ and show that they can be decomposed into simpler problems. The multivariable problem in (\ref{Pnew}) is a non-linear and non-convex optimization problem. Following a similar approach as in [ref-MEC], the objective function $f_0$ for $(P_{\text{int}})$ is a convex function of $s_i$. Furthermore, provided that the gradient of $f_0(\cdot)$ with respect to $s_i$ evaluated at $s_i = 0$ is positive, which is often satisfied in typical network settings, then the total energy in problem $(P_{\text{int}})$ is an increasing function of each $s_i$ and there exists an optimal point, $s_i^\star \ \forall i \in [1,K]$, which minimizes $E_{\text{total}}$ within the latency constraint. If offloaded data $\boldsymbol{s}$ is fixed, then problem $(P_{\text{int}})$ turns out to be convex in the remaining variables as stated in the following lemma. Lemma \ref{lemma1} lets us decompose the original non-convex problem $(P_{\text{int}})$ into simpler convex subproblems which will be used in the subsequent algorithm design. 
\begin{lemma}\label{lemma1}
For a given set of offloaded data $\boldsymbol{s}$, the problem $(P_{\text{\emph{int}}})$ is convex in the remaining variables $\boldsymbol{t, W_q}$.
\end{lemma}
\begin{proof}
Proof follows by examining each constraint and showing that with fixed $s_i$, it is a convex function. Details in Appendix A.
\end{proof}

Since CPU frequencies are not optimizing variables, for a given value of the offloaded data $s_i$, the computation time for the offloaded data $T_2$ can be pre-determined in closed form directly from constraint (g) in (\ref{Pnew}) and hence constraint (\ref{Pnew}g) can be excluded from the problem $(P_{\text{int}})$. Also, at a fixed value of $\boldsymbol{s}$, considering wireless charging as an opportunistic feature in addition to computation offloading, $(P_{\text{int}})$ is also separable in $\boldsymbol{t}$ and $\boldsymbol{W_q}$ as stated in the lemma below. 

\begin{lemma}\label{sep_prob}
Given that wireless charging is opportunistic, at a fixed value of $s_i$, problem $(P_{\text{int}})$ is separable in terms of variable $\boldsymbol{t}$ and $\boldsymbol{W_q}$ as follows
\begin{align}\label{Pdecomp}
&(P2) \ \ \ \ \min_{\boldsymbol{t}} \ \ \ (1-w) E_u + w E_{m1} \ \ \text{s.t.}\ \ (\ref{Pnew}c) - (\ref{Pnew}f) \\
&(P3)\ \  \ \ \min_{\boldsymbol{W_q}}  \ \ \ w E_{m2} \ \ \text{s.t.}\ \  (\ref{Pnew}h) - (\ref{Pnew}j)   
\end{align}
\end{lemma}
\begin{proof}
The objective function of $(P_{\text{int}})$ in (\ref{Pnew}) constitutes of distinct component terms dependent on the time allocation variables $t_{u,i}, t_{d,i}$ and the transmit covariance matrix $\boldsymbol{W_q}$. More specifically, the MEC energy consumption $E_m$ can be decomposed as $E_m = E_{m1} + E_{m2}$, where $E_{m1} = \sum_{i=1}^{K} \Big[\frac{t_{d,i}(2^{\frac{\mu s_i}{t_{d,i}B}} - 1)\Gamma_{2}\sigma_{2,i}^2}{N \gamma_{i}}  + \kappa_m d_m f_{mi}^2 s_i \Big]$ is the energy consumed for computation and transmission and $E_{m2} = (T_d - T_1 - T_3) \text{tr}(\boldsymbol{W_q})$ is the energy consumed for wireless charging. 

Problems (P2) and (P3) minimize the MEC energy consumption for wireless charging (P3) and computation offloading (P2) separately. The only variable coupling these two problems is $T_c$ which must satisfy $T_c = T_d - T_1 - T_3$. Since wireless charging is an opportunistic feature of the system model with priority given to data computation, such that wireless charging takes place during the time which remains after data transmissions, (P2) can first be solved to find the optimal $T_1^\star$ and $T_3^\star$ for computation offload and results download respectively. The charging time constraint (\ref{Pnew}h) in (P3) can then be obtained as $T_c = T_d - T_1^\star - T_3^\star$ which becomes a fixed value and is no longer an optimization variable in (P3). Based on this reasoning, for any of the three modes of operation, the problem $(P_{\text{int}})$ is separable. 
\end{proof}

\subsection{Analysis for Computation Offloading Sub-problem}
Next we present the solution for the optimal time allocation for the computation offloading problem (P2). Since the problem is convex based on Lemma \ref{lemma1}, we adopt a primal-dual solution using the Lagrangian duality analysis similar to that proposed in  \cite[Theorem 1]{Malik2020} and derive the optimal solution as given in Theorem~\ref{theorem2} below.

\begin{theorem}\label{theorem2}
The offloading and downloading time, $t_{u,i}$ and $t_{d,i}$ respectively, can be obtained as a solution of the form 
\begin{equation}\label{LambertSol}
x = \frac{c B}{\ln 2} \Big(W_0 \Big(\frac{-y}{\sigma^2 e} - \frac{1}{e}\Big) + 1 \Big)
\end{equation}
where $y = -\frac{\beta_i + \xi_i}{(1 - w)}$, $x = x_{1,i} = \frac{1}{t_{u,i}}$, $c = \frac{\nu}{s_i}$, $\sigma^2 = \frac{ \Gamma_{1} \sigma_{1,i}^2}{N \gamma_i}$ to solve for $t_{u,i}$, and  $y = \frac{- \phi_i}{w}$, $x = x_{2,i} = \frac{1}{t_{d,i}}$, $c = 1/\mu s_i$, and  $\sigma^2 = \frac{\Gamma_{2} \sigma_{2,i}^2}{N \gamma_i}$ to solve for $t_d,i$. Here $\xi_i$, $\beta_i$ and $\phi_i$ are the dual variables associated with the constraints (d), (e) and (g) of problem $(P_{\text{int}})$ in (\ref{Pnew}) respectively.
\end{theorem}
\begin{proof}
The solution in (\ref{LambertSol}) can be obtained directly by applying KKT conditions on the Lagrangian dual of the problem (P2) or $P_{\text{seq,CO}}$ with respect to $t_{u,i}$ and $t_{d,i}$. Detailed proof can be obtained using an approach similar to that in  \cite[Theorem 1]{Malik2020} and is omitted for brevity.
\end{proof}

\section{Energy Beamforming for Wireless Charging}
The formulated sub-problem (P3) below aims at minimizing the energy consumed for power transfer during the wireless charging operation while making sure the charged amount is at least an $\alpha_i$ proportion of the requested amount, whereas $0 < \alpha_i \leq 1$. We re-write (P3) with relevant constraints from (\ref{Pnew}) as follows
\setcounter{equation}{22}
\begin{align*}\label{WP}
(P3): \ \min_{\boldsymbol{W_q}} \ \ &T_c \text{tr}(\boldsymbol{W_q}) \tag{\theequation}\\
\text{s.t.} \ \ &\text{tr}(\boldsymbol{W_q}) \leq P \tag{a}\\
&\xi \text{tr}(h_i^\ast \boldsymbol{W_q} h_i) T_c \geq \alpha_i e_i \ \ \forall i = 1...K \tag{b}
\end{align*}
\textcolor{blue}{For the presented problem ($P3$), a straightforward approach to write constraint (b) would be as follows
\begin{equation}\label{charge_max}
\xi \text{tr} (\boldsymbol{h_i^\ast W_q h_i})T_c - \alpha_i e_i  \leq 0
\end{equation}
with $\alpha_i = 1 \ \forall i$ to represent that the received energy is always less than the requested energy. This is infact how the constraint is depicted in the sequential problem approach in [Malik-maxCharging]. A key difference in ($P3$) presented here, and ($P_{WC}$) in [Malik-maxCharging] is the objective function. While the objective here is to minimize the charging energy, the objective for desigining the energy beamforming matrix in [Malik-maxCharging] is to maximize the received energy, without minimizing the charging energy. In the energy minimization problem $P3$ above, writing constraint (b) as in (\ref{charge_max}) would render the optimal beamforming matrix, $\boldsymbol{W_q^\star = 0}$. Therefore, $P3$ introduces a best effort approach towards wireless charging, such that the energy delivered to the users is maximized while also minimizing the overall energy consumption. In this regard we thus present the constraint on the received energy as in (b) by using the auxiliary variable $\alpha_i$ to ensure that the received energy is less than the requested energy, but $\boldsymbol{W_q^\star \neq 0}$.}

Problem (P3) is a semi-definite programming with linear objective function and linear constraints and hence is convex. We can show that strong duality holds since Slater's condition is satisfied, that is, we can find a strictly feasible point ($\boldsymbol{W_q} = p\boldsymbol{I}_{N \times N}$, $p \leq P/N$, $0 \leq \alpha_i \leq 1 \ \forall i$) in the relative interior of the domain of the problem where the inequality constraints hold with strict inequalities~\cite{Boyd2004}. 

From the definition of charging time, as $T_c = T_d - T_1 - T_3$, the problem (P3) has an interdependency on the optimization problem (P2). However, based on Lemma \ref{sep_prob}, since (P2) and (P3) are separable, we can use the optimal time allocation obtained as a solution of (P2) to find the energy beamforming matrix, $\boldsymbol{W_q}$ in (P3). Theorem \ref{theorem1} below provides the optimal beam directions for wireless charging.

\begin{theorem}\label{theorem1}
Let the eigenvalue decomposition of the optimal energy beamforming matrix be $\boldsymbol{W_q}^\star =\boldsymbol{U_q \Lambda_q^\star U_q^\ast}$, where $\boldsymbol{U_q} \in \mathbb{R}^{N \times N}$ defines the directions of energy beams and diagonal $\boldsymbol{ \Lambda_q^\star }$ is the beam power allocation matrix. Then the optimal directions for energy beams are $\boldsymbol{U_q^\star} = \boldsymbol{U_B}$, where $\boldsymbol{U_B}$ is obtained from the eigenvalue decomposition of $\boldsymbol{B} = \boldsymbol{U_B \Lambda_B U_B^\ast}$, such that $\lambda_{B,1} \leq \lambda_{B,2} \leq \ldots \leq \lambda_{B,N}$, where
\begin{equation*}
\boldsymbol{B} = (T_c + \lambda_5) \boldsymbol{I} -  \xi T_c \sum_{i=1}^{K} \psi_i \alpha_i \boldsymbol{h_i h_i^\ast}
\end{equation*}
Here $\lambda_5$ and $\psi_i$ are the dual variables associated with constraint (\ref{WP}a) and the $i^{\text{th}}$ constraint in (\ref{WP}b) respectively.
\end{theorem}
\begin{proof}
See Appendix B.
\end{proof}

Theorem \ref{theorem1} provides the optimal directions of the energy beams for the beamforming matrix, $\boldsymbol{W_q}$. What is left now is to obtain the optimal power allocation across the energy beams, that is, the eigenvalues of the transmit covariance matrix for wireless charging. To this end, we substitute the optimal beam directions from Theorem \ref{theorem1} into (P3) and re-write the formulation in terms of the beam power allocation only as (P4) below. Beam power allocation, $\boldsymbol{\lambda_q}$, can then be obtained as a solution to a Linear Programming (LP) problem given in Theorem \ref{theorem3} below.
\begin{theorem}\label{theorem3}
The optimal beam power allocation is derived by solving the LP problem below
\setcounter{equation}{23}
\begin{align*}\label{P4}
(P4): \ \min_{\boldsymbol{\lambda_q}} \ \ &\sum_{i=1}^K \lambda_{q,i} \tag{\theequation}\\
\text{s.t.} \ \ & \sum_{i=1}^K \lambda_{q,i}  \leq P, \ \ \ \  \lambda_{q,1} \geq ... \geq \lambda_{q,K} \geq 0 \tag{a-b}\\
&\boldsymbol{A \lambda_q} \geq \text{\emph{diag}}(\boldsymbol{\alpha}) \boldsymbol{b} \tag{c}
\end{align*}
where $\boldsymbol{\lambda_q} = [\lambda_{q,1}, ..., \lambda_{q,K}]^T$, $\boldsymbol{A} \in \mathbb{R}^{K \times K} = [\boldsymbol{a_1^\ast}...\boldsymbol{a_K^\ast}]$, $\boldsymbol{a_i}^\ast = \text{\emph{diag}}(\boldsymbol{q_i q_i^\ast})$, $\boldsymbol{q}_i^\ast = \boldsymbol{h}_i^\ast  \boldsymbol{U_B}$ and $\boldsymbol{b} \in \mathbb{R}^{K \times 1} = [\pi_1 ... \pi_K]$, $\pi_i = \frac{e_i}{\xi T_c} \ \forall i = 1...K$.
\end{theorem}
\begin{proof}
See Appendix B.
\end{proof}

In problem (P4), $\boldsymbol{\alpha} \in [0,1]$ is an auxiliary variable to ensure that the amount of charged energy is as large as possible within the transmit power constraint. While different values of $\boldsymbol{\alpha}$ will result in different power allocation, we are interested in the largest $\boldsymbol{\alpha}$ that makes (P4) feasible, so that the amount of received energy is largest while minimizing the transmit power. Since the optimal solution of the LP is linear to the constraint $P$, we can solve this problem without loss of optimality by first setting $\boldsymbol{\alpha = 1}$ and removing the power constraint, then solve for the resulting LP. If the sum of solved $\lambda_{q,i}$ is more than $P$, then $\boldsymbol{\alpha}$ will be the scaling factor to bring this sum to be equal to $P$ and all optimal values of $\lambda_{q,i}$ will be scaled by $\alpha_i$. Otherwise $\alpha_i$ stays as 1 and $\lambda_{q,i}$ stays unchanged. The largest energy ratio variable $\boldsymbol{\alpha}$ is hence obtained as
\begin{equation}\label{alpha}
  \boldsymbol{\alpha} =\begin{cases}
    \boldsymbol{1}, & \text{$\sum_{i=1}^{K} \lambda_{q,i}^{\star,0} <= P$}\\
    \text{diag}(\boldsymbol{b})^{-1}\boldsymbol{A\lambda}_q^{\star,0} \frac{P}{\sum_{i=1}^K \lambda_{q,i}} , & \text{otherwise}
  \end{cases}
\end{equation} 
where $\boldsymbol{\lambda}_q^{\star,0}$ are the optimal values of $\lambda_{q,i}$ when $\alpha_i = 1$ and transmit sum power constraint $P$ is removed.

\section{Nested Algorithm Design}
\textcolor{blue}{While problem ($P_{\text{int}}$) is not convex in all the optimizing variables, Lemma \ref{lemma1} shows that by fixing the offloaded bits s, the problem is convex in all the remaining optimizing variables with
a convex objective function and a convex feasible set. We divide problem ($P_{\text{int}}$) into an outer iterative problem which solves for the optimal balance between offloaded bits and those retained at the users and two inner subproblems described in Lemma \ref{sep_prob}: problem ($P2$) which finds the optimal time allocation given a fixed number of offloaded bits $\boldsymbol{s}$, while ($P3$) solves for the transmit covariance matrix for power transfer. Since ($P_{\text{int}}$) is jointly convex in $\boldsymbol{t}$ and $\boldsymbol{W_q}$, any algorithm which solves a convex problem can be applied, however, standard convex-solvers are often inefficient due to their inability to exploit the specific problem structure. We therefore propose a customized nested algorithm, as follows, which includes as outer algorithm to determine $\boldsymbol{s^\star}$ and an inner algorithm to solve for $\boldsymbol{t^\star}$ and $\boldsymbol{W_q^\star}$ to efficiently reach the solution for the problem ($P_{\text{int}}$).}
\subsection{Nested Algorithm Architecture}
Based on Lemma \ref{lemma1}, the algorithm for solving $(P_{\text{int}})$ is designed to have a nested architecture with an outer and an inner loop, in which the outer loop solves for $s_i$ decrementally while the inner loop solves for the remaining variables at a fixed value of $s_i$. More specifically, the nested algorithms work as follows. We first initialize the offloaded bits $\boldsymbol{s}$ and the dual variables in the outer algorithm. At the current value of $s$, the inner algorithm is executed, for which we use a primal-dual approach employing a subgradient method. At convergence where the stopping criterion for the dual problem is satisfied, the inner algorithm returns the control to the outer algorithm. Based on the newly updated primal solution from the inner algorithm, we proceed to updating $\boldsymbol{s}$ by some $\Delta s_i$ for each user for the next iteration of the outer algorithm, using a latency aware descent algorithm. Similar to [ref-MEC], the latency aware descent algorithm is based on the standard Newton method with a novel modification to the classical stopping criterion to account for the latency constraint.

\begin{figure}[t]
\centering
\includegraphics[scale = 0.6]{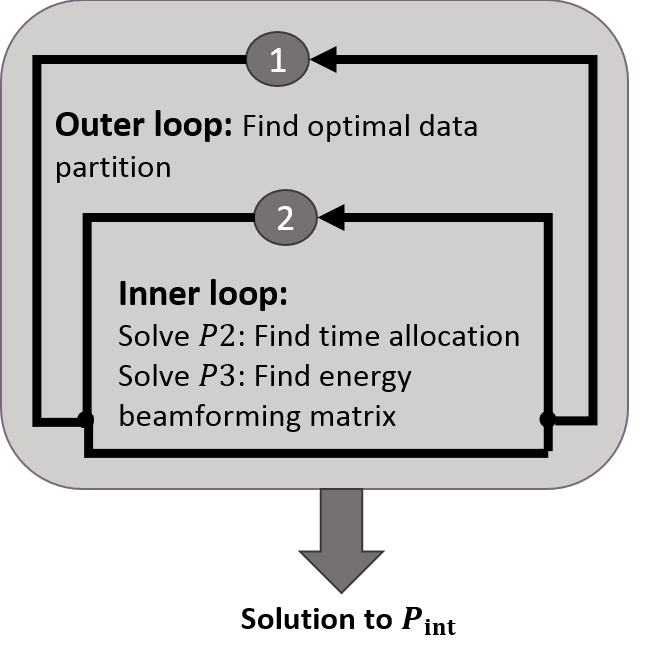}
\caption{Algorithm architecture for $P_{int}$}
\label{sol_flow}
\end{figure}

Thus, for each iteration of the outer algorithm, we solve for the inner optimization problems (P2) and (P3) in sequence. For a given value of $\boldsymbol{s}$, we solve (P2) to obtain time allocation, and calculate the charging time $T_c = T_d - T_1 - T_3$ which is used by problem (P3) to find the optimal energy beamforming matrix. These steps for the nested optimization are repeated until a minimum point for the weighted total energy consumption is reached where all the constraints in the original problem $(P_{\text{int}})$ are satisfied. At each iteration of (P3), we solve the LP problem (P4) to find the optimal beam power allocation using a standard convex solver. The algorithm flow is depicted in Figure \ref{sol_flow} and the steps for solving $(P_{\text{int}})$ are given in Algorithm 1.

\textcolor{blue}{A nested algorithm approach has also been proposed in [Malik-maxCharging] to find the optimal data partitioning, time allocation and energy beamforming, however there are subtle differences in the two constructions which must be brought to notice. For the integrated formulation presented in this paper, the solution for energy beamforming is also a part of the inner algorithm as the subproblem (P3). On the other hand, [Malik-maxCharging] constructs two sequential problems $(P_{\text{WC}})$ to solve for $\boldsymbol{W_q}$m and $(P_{\text{CO}})$ which solves for the optimal $\boldsymbol{s}$ and $\boldsymbol{t}$. Different from this work, the sequential problem, $(P_{\text{WC}})$ in [Malik-maxCharging] solves for the transmit covariance matrix $\boldsymbol{W_q}$ as an independent problem after obtaining the optimal time allocation solution from $(P_{\text{CO}})$. The optimal $\boldsymbol{t^\star}$ from $(P_{\text{CO}})$ is used to calculate the charging time $T_c$. There is a subtle difference in that for $(P_{\text{int}})$, at each iteration of the outer algorithm, the two inner problems (P2) and (P3) are solved sequentially but the value for charging time $T_c$ is only the temporary value that is optimal for a given value of $\boldsymbol{s}$ and is not the final optimal charging time. For the sequential problem, however, the wireless charging problem $(P_{\text{seq,WC}})$ is solved after obtaining the final optimal solution for time allocation from $(P_{\text{CO}})$ which is then used to obtain the optimal charging time $T_c^\star$ to find the optimal energy beamforming matrix.}
\begin{algorithm}[t]
\caption{Solution for ($P_{\text{int}}$)}
\text{Given:} Distances $d_{i} \ \forall i$. Channel $\boldsymbol{H = G^{T}}$. Precision, $\epsilon_1, \epsilon_2$, Data $u_i$, Latency $T_d$. \text{Initialize:} $s_i$\\
\textbf{Begin Outer Algorithm for $P_{\text{int}}$}\\
\textbf{Given} a starting point $\boldsymbol{s}$, \textbf{Repeat}
\begin{enumerate}
\item Compute $\Delta \boldsymbol{s}$ and initialize dual variables, $\lambda_1,\lambda_5, \beta_i, \xi_i, \phi_i, \psi_i \forall i$.
\item \textbf{Begin Inner Algorithm for (P2)}
\begin{itemize}[leftmargin=*]
\item Calculate $t_{u,i}$ and $t_{d,i}$, using (\ref{LambertSol}). Then $T_1^\star = \max t_{u,i}^\star$ and $T_3^\star = \max t_{d,i}^\star$.
\item Update $p_i$ and $\eta_i$ using (\ref{power_alloc}) and calculate $\sigma_{1,i}^2$ and $\sigma_{2,i}^2$.
\item Find dual function in (\ref{DualFn}), stop if dual variables converge with $\epsilon_2$, else find subgradients in (\ref{subgrads}-d), update dual-variables using subgradient method and continue
\end{itemize}
\textbf{End Inner Algorithm for (P2)}
\item \textbf{Begin Inner Algorithm for (P3)}
\begin{itemize}
\item Calculate time for wireless charging, $T_c^\star = T_d - T_1^\star - T_3^\star$
\item Find $\boldsymbol{\lambda_q}^\star$ from (P4). Then $\boldsymbol{W_q}^\star = \boldsymbol{U_B \Lambda_q^\star U_B^\ast}$, where $\boldsymbol{\Lambda_q^\star} = \textbf{diag}(\boldsymbol{\lambda_q}^\star)$
\item Find dual function in (\ref{P2dualfn}), stop if dual variables converge with $\epsilon_2$, else, find subgradients in (\ref{subgrads}), update dual-variables using subgradient method and continue
\end{itemize}
\textbf{End Inner Algorithm for (P3)}
\item \textit{Line search and Update}. $s_i := s_i + t_i\Delta s_i$.
\end{enumerate}
\textbf{Until} stopping criterion is satisfied with $\epsilon_1$ or latency constraint $T_d$ is met.\\
\textbf{End Outer Algorithm for $P_{\text{int}}$}\\
\end{algorithm}

\subsection{Inner Primal-Dual Algorithms}
For the inner algorithm, we design a primal-dual algorithm where the primal variable values are obtained as closed form functions of the dual variables, and the dual variables are found by solving the dual problem using a sub-gradient methods. The dual-function for the convex optimization problem (P2) can be defined as
\begin{equation}\label{DualFn}
    g_{P2}(\lambda_1,\boldsymbol{\beta, \xi, \phi}) = \inf_{\boldsymbol{t}} \mathcal{L}_{P2}(\boldsymbol{t}, \lambda_1, \boldsymbol{\beta, \xi, \phi})
\end{equation}
where $\mathcal{L}_{P2}$ is the Lagrangian for problem (P2) and the dual-problem is defined as
\begin{align}\label{PDual}
    \text{P2-dual: }\max \ &g_{P2}(\lambda_1,\boldsymbol{\beta, \xi, \phi})  \ \  \text{s.t. } \lambda_1 \geq 0, \beta_i, \xi_i, \phi_i \geq 0 \ \forall i = 1...K 
\end{align}
where $\lambda_1$, $\boldsymbol{\beta, \xi}$, and $\boldsymbol{\phi}$ are the dual variables associated with constraints (c-f) in (\ref{Pnew}), respectively.

Similarly, the dual-function for the convex optimization problem (P3) is defined in (\ref{P2dualfn}). The dual problem is given as
\begin{align}\label{P2dual}
 \text{P3-dual: }\: \max \ \ &g_{P3}(\boldsymbol{\psi}, \lambda_5) \ \  \text{s.t.} \ \ \lambda_5 \geq 0, \psi_i \geq 0 \ \text{for } i = 1...K 
 \end{align} 
Using the closed form expressions for the primal variables in terms of the dual-variables as in Theorems \ref{theorem2}-\ref{theorem3}, the dual functions above are functions of only the dual-variables. The above dual problems can then be solved using the subgradient method~\cite{Boyd2003}. 

The subgradient terms with respect to all dual variables of the original problem (P2), (P3) and $(P_{\text{seq,WC}})$ are as given below
\setcounter{equation}{30}
\begin{align*}\label{subgrads}
&\nabla_{\lambda_1}\mathcal{L} = \sum_{j=1}^{3} T_j - T_{\text{delay}} \tag{\theequation a}\\
&\nabla_{\beta_i}\mathcal{L} = t_{u,i} - T_1, \ \ \nabla_{\phi_i}\mathcal{L} = t_{d,i} - T_3 \ \ \nabla_{\xi_i}\mathcal{L} = \frac{c_i q_i}{f_{u,i}}  + t_{u,i} - T_d, \ \ \ \ \ \  \tag{\theequation b-d}\\
&\nabla_{\psi_i}\mathcal{L} = \alpha_i e_i - \xi \text{tr} \left( \boldsymbol{h_i^\ast W_q h_i} \right )T_c  \ \ \ \nabla_{\lambda_5}\mathcal{L} = \text{tr}(\boldsymbol{W_q}) - P \tag{\theequation e-f}
\end{align*}

For implementation of the inner algorithms, we use the subgradient method to solve the constrained convex optimization problems (P2) and (P3)~\cite{Boyd2003}. The designed algorithms find the subgradient for the negative dual functions ($-g_{P2}, -g_{P3}$) since the dual problems in (\ref{P2dual}, \ref{PDual}) are maximization problems for the respective dual functions. At each iteration, the primal variables are updated based on Theorems \ref{theorem2}-\ref{theorem3}. The dual variables vector $x$ is updated as $\boldsymbol{x}^{(k+1)} = \boldsymbol{x}^{(k)} - \beta_k \boldsymbol{g}^{(k)}$, where $\beta_k$ is the $k^{\text{th}}$ step-size, and $\boldsymbol{g}^{(k)}$ is the subgradient vector at the $k^{\text{th}}$ iteration evaluated using the sub-gradient expressions in (\ref{subgrads}-f). We use the non-summable diminishing step size, setting $\beta_k = 1/\sqrt{k}$, using which the algorithm is guaranteed to converge to the optimal value~\cite{Boyd2003}. Since the subgradient method is not a descent method, the algorithms keep track of the best point for the dual functions at each iteration of the inner algorithm. These primal-dual update steps are repeated until the desired level of precision, $\epsilon_2$, is reached for the stopping criterion.

For the diminishing step size as that considered, the subgradient method is guaranteed to converge~\cite{Boyd2003} as $k \to \infty$. In the subgradient method, since the key quantity is not the function value but rather the Euclidean distance to the optimal set~\cite{Boyd2003}, therefore, for our implementation we define the stopping criterion as: $\lVert \boldsymbol{g}^{(k+1)} - \boldsymbol{g}^{(k)} \rVert_2 \leq \epsilon_2$. The steps for the integrated algorithm are shown in Algorithm 1.
\section{Numerical Results}
\begin{figure}[t]
\centering
\includegraphics[scale = 0.6]{layout}
\caption{System layout}
\label{system_model}
\end{figure}
In this section, we evaluate the solution of energy minimization problem (P) with respect to energy and time consumption, the partition of bits offloaded to the MEC for computation and the received energy via wireless charging. We consider a $20\text{m}\times20\text{m}$ area with 4 APs and 16 users randomly located with $K = 4$ users per AP's coverage area and $N = 100$ as shown in Figure~\ref{system_model}. For simulations, $w = 10^{-3}$, $T_d = 20$ms (for AR/VR applications~\cite{Intel2017}), $B = 5$MHz, $\tau_c = B T_d$, $\Gamma_1 = \Gamma_2 = 1.25$, $\mu = 2$, $\kappa_i = 0.5$pF, $\kappa_m = 5$pF, $c_i = 1000$, $d_m = 500$, $\gamma = 2.2$, $\sigma = 2.7$dB, $\sigma_r^2 = -127$dBm, $\sigma_k^2 = -122$dBm, $f_{u,i} = f_u = 1800$ MHz $\forall i$. Each MEC processor has 24 cores with maximum frequency of $3.4$GHz, and we use $f_{m,i} = f_m = \frac{24 \times 3400}{K}$ MHz $\forall i$. Transmit power available at user and AP is 23 dBm and 46 dBm respectively. To calculate the interference and noise power ($\sigma_{1,i}^2$, $\sigma_{2,i}^2$) which include massive MIMO pilot contamination and intercell interference, we assume that user terminals transmit at their maximum power, that is $p_{qi} = 23$dBm, and the interfering APs use equal power allocation in the downlink, that is $\eta_{qi} = \frac{1}{K} \ \forall i$. Numerical results are averaged independent channel realizations of $\mathbf{H}$ and $\mathbf{G}$ for 1000 spatial realizations (randomly generated user locations).

\subsection{Comparison of Wireless Charging Schemes}
\begin{figure}[t]
\centering
\includegraphics[scale = 0.6]{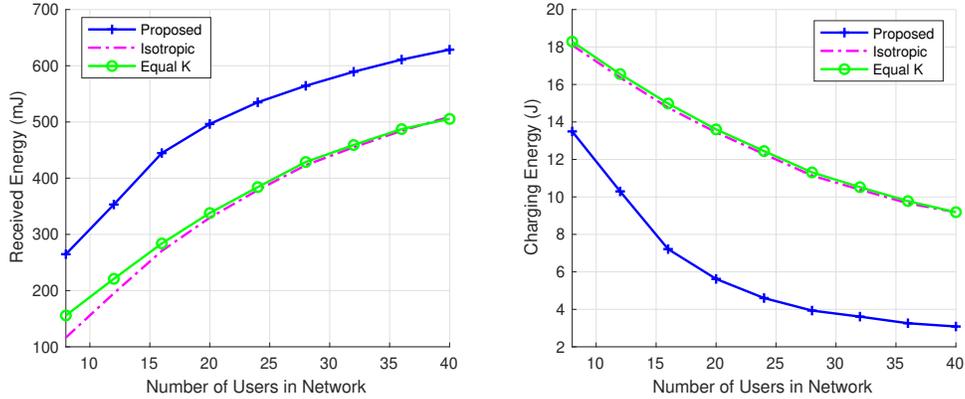}
\caption{Comparison of the proposed integrated and sequential energy beamforming schemes with isotropic wireless charging, and directed K-beam charging with equal power allocation (equal K)}
\label{beam_effect}
\end{figure}
Figure \ref{beam_effect} shows a comparison of the proposed integrated and sequential wireless charging schemes with two other schemes: (i) isotropic scheme where $\boldsymbol{W_q} = \frac{P}{N}\boldsymbol{I}$ and equal charging power $P/N$ is allocated across all $N$ antennas of the AP, and (ii) equal K with directional charging using the beamforming directions proposed in Theorem \ref{theorem1}, but with equal power allocation $P/K$ across $K$ energy beams. For fairness of comparison with the sequential scheme, we use power scaling for the other two schemes such that the users receive energy at most equal to the requested amounts similar to the sequential scheme. Since wireless charging is proposed as a billable service for future networks, this is also a necessary design consideration from the service providers' and consumers' perspective. 

Figure \ref{beam_effect} shows the received energy on the left and the transmitted energy on the right. As illustrated in this figure, the sum received energy for the proposed scheme is significantly larger than the other two schemes. The wireless charging performance for the isotropic and beamforming with equal power allocation scheme are similar. However, for smaller networks the equal power allocation scheme with directed power transfer does offer some improvement over the isotropic scheme in terms of the received energy. The proposed integrated charging energy minimization scheme consumes the lowest charging energy and offers better received energy performance than both the isotropic and equal power schemes nonetheless. 

\subsection{Number Of Charging Energy Beams}
\begin{figure}[t]
\centering
\includegraphics[scale = 0.6]{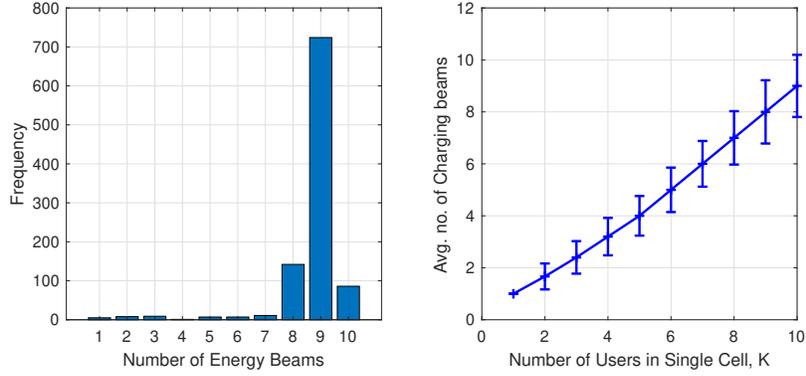}
\caption{Number of charging energy beams for the sequential and integrated approaches for K = 10 users.}
\label{num_beams}
\end{figure}
Another interesting finding presented in the plot (bottom) in Figure \ref{num_beams} is the optimal number of energy beams for $K = 10$ users per cell. For the isotropic wireless charging, there are always $N > K$ energy beams. For the case of $K$ beams with equal power allocation, the number of beams is equal to the number of users in the cell. While multiple energy beams may be necessary for a multi-user system as also previously discussed in~\cite{Zeng2017}, the optimal number of energy beams for the integrated charging energy minimization scheme is usually less than the number of users. Since each energy beam can contribute as additional RF charging sources for neighboring users, the transmit beamforming can be intelligently designed as proposed to limit the number of energy beams which can prevent energy losses caused by transmitting energy in numerous directions and hence also contribute to energy minimization. 

\subsection{Charging Efficiency}
\begin{figure}[t]
\centering
\includegraphics[scale = 0.6]{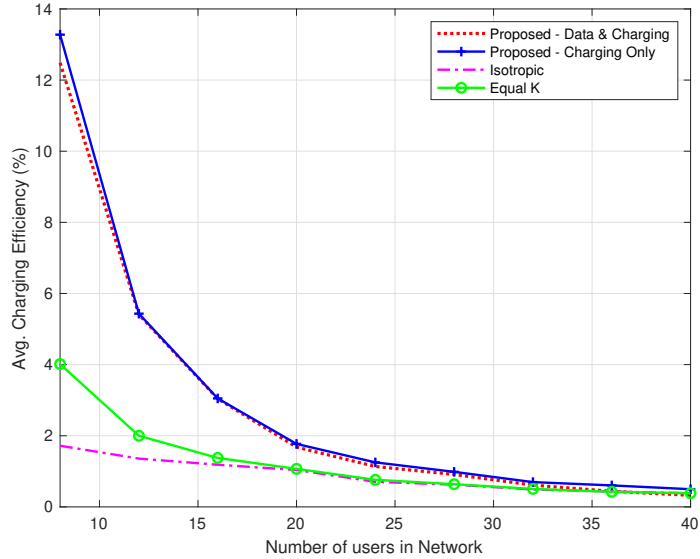}
\caption{Average charging efficiency}
\label{efficiency}
\end{figure}
Figure \ref{efficiency} shows the average charging efficiency, Average charging efficiency is defined as the average percentage of received energy at the users end compared to the requested energy. The figure shows a comparison for the directed equal power and isotropic schemes to the proposed scheme under two operating modes, the \text{charging only} mode where connected users request wireless charging but do not require computation at the edge, and the \textit{data and charging} mode where connected users request both wireless charging as well as computation offloading. The average efficiency is seen to decrease with an increase in the network size as expected. While the directed beams equal charging scheme shows some improvement over the isotropic scheme for small network sizes, both schemes generally have less than 50\% efficiency compared to the proposed scheme. As illustrated, even under the joint data and charging mode, where the MEC-AP simultaneously optimizes the resources required for both computation offloading as well as wireless charging, the decrease in charging efficiency is negligible compared to the charging only mode. 

\section{Conclusion}
In this paper, we examined a massive MIMO enabled multi-access edge computing network providing computation offloading and on-demand wireless charging to its connected users under a round trip latency constraint. We formulated a novel system-level problem to minimize the total transmit energy consumption. We design an efficient nested algorithm by an optimal division into convex subproblems to solve for data partitioning, time allocation and transmit energy beamforming matrices. Our algorithm demonstrates that even with significant amounts of data to be computed, the network can deliver decent amounts of charged energy to the users, therefore validating a practical coexistence of computation offloading and wireless charging. A comparison with isotropic power transfer and equal power energy beamforming shows that optimal design of the energy beamforming for wireless charging is significantly more energy efficient, and is necessary for implementing user fairness.


\section{Appendix}
\subsection{Appendix A - Proof for Lemma \ref{lemma1}}
Consider problem $(P_{\text{int}})$ in (\ref{Pnew}) at fixed values of $s_i$. The objective function is affine and convex. Convexity in $\boldsymbol{t}$ can be established similar to [lemma1-refMEC]. Constraint (b) contains a function of the form $f_1 = (T_d - T_1 - T_3) \text{tr}(\boldsymbol{W_q})$, with affine term $T_d \text{tr}(\boldsymbol{W_q})$. The second and third term are of similar form. Considering $\tilde{f}_1 = -T_1 \text{tr}(\boldsymbol{W_q})$, to check for joint convexity in $T_1$ and $\boldsymbol{W_q}$, the Hessian of $\tilde{f}_1$ is the block matrix, $\nabla^2 \tilde{f}_1 = [\boldsymbol{0}_{N \times N} \ \ -\boldsymbol{I}_{N \times N}; -\boldsymbol{I}_{N \times N} \ \ \boldsymbol{0}_{N \times N}]$, with repeated eigenvalues $\pm 1$ and therefore doesn't show convexity. However, the sublevel sets $\{(T_1 \in \mathbb{R}^+, \boldsymbol{W_q} \in \mathbb{R}^{N\times N}), -T_1 \text{tr}(\boldsymbol{W_q}) \leq \alpha \}$ are jointly convex in $T_1$ and $\boldsymbol{W_q}$ in the domain of the function and therefore the function $\tilde{f}_1$ is quasiconvex~\cite[Example 3.31]{Boyd2004}. Therefore constraint (\ref{Pnew}b) is a sum of convex and quasiconvex functions with convex sets and sublevel sets respectively. Similarly, constraint (i) also has convex sublevel sets with a quasiconvex function of the form  $-T_c \text{tr} (\boldsymbol{h_i^\ast W_q h_i})$.  Constraints (j) is linear trace function of $\boldsymbol{W_q}$ and hence convex. 

Based on the above, the objective is convex and all the constraints are either convex or have convex level sets in the remaining variables. Thus the problem is convex at given $s_i$.

\subsection{Appendix B - Proof for Theorem \ref{theorem1} and \ref{theorem3}}
\subsubsection{Proof for Theorem \ref{theorem1}}
Form matrix $\boldsymbol{B}$ from the Lagrangian for problem (P3) as follows
\setcounter{equation}{31}
\begin{align*}\label{Lsub}
\mathcal{L}_{P3} &= T_c \text{tr}(\boldsymbol{W_q})+ \lambda_5 \left(\text{tr}(\boldsymbol{W_q}) - P \right) - \xi T_c \text{tr} \left(\left (\sum_{i=1}^{K} \psi_i h_i h_i^\ast  \right)\boldsymbol{W_q} \right) + \sum_{i=1}^{K} \psi_i \alpha_i e_i \\
&= \text{tr} \left( \left [ (T_c + \lambda_5) \boldsymbol{I} -  \xi T_c \sum_{i=1}^{K} \psi_i \boldsymbol{h_i h_i^\ast}\right ]\boldsymbol{W_q}\right) + \sum_{i=1}^{K} \psi_i \alpha_i e_i - \lambda_5 P \\
&= \text{tr} \left( \boldsymbol{BW_q}\right) + \sum_{i=1}^{K} \psi_i \alpha_i e_i - \lambda_5 P \tag{\theequation}
\end{align*}
The dual-function for the problem (P3) is obtained as
\begin{equation}\label{P2dualfn}
g_{P3}(\boldsymbol{\psi}, \lambda_5) = \min_{\boldsymbol{W_q}} \mathcal{L}_{P3} (\boldsymbol{W_q, \psi},\lambda_5)
\end{equation}
To minimize the Lagrangian in (\ref{Lsub}) for the dual function in (\ref{P2dualfn}), we only need to consider the term involving $\boldsymbol{W_q}$  
\begin{equation}\label{partL}
\min_{\boldsymbol{W_q}} \ \ \text{tr} \left( \boldsymbol{BW_q}\right) 
\end{equation}
By applying an inequality relating the trace of a matrix product to the sum of eigenvalue products~\cite[Ch.~9, H.1.h.]{Olkin1979}, $\text{tr}(\boldsymbol{BW_q})$ is minimized by choosing $\boldsymbol{U_q = U_B}$ such that
\begin{equation}\label{trineq}
\text{tr}(\boldsymbol{BW_q}) = \sum_{i = 1}^{N} \lambda_{B,i} \cdot \lambda_{q,i}
\end{equation}
where the eigenvalues of $\boldsymbol{W_q}$ are in descending order $\lambda_{q,1} \geq \lambda_{q,2} \geq \ldots \geq \lambda_{q,N}$ and those of matrix $\boldsymbol{B}$ are in ascending order such that $\lambda_{B,1} \leq \lambda_{B,2} \leq \ldots \leq \lambda_{B,N}$ and the eigenvectors $\boldsymbol{U_B}$ are obtained based on this order of the corresponding eigenvalues in $\boldsymbol{\Lambda_B} = \text{diag}(\boldsymbol{\lambda_B})$. Since the eigenvalues of $\boldsymbol{B}$ and $\boldsymbol{W_q}$ are in reverse order to each other, the sum of their eigenvalue products yields the minimum value for $\text{tr}(\boldsymbol{BW_q})$ in (\ref{trineq}).

\subsubsection{Proof for Theorem \ref{theorem3}}
In the eigenvalue decomposition of $\boldsymbol{W_q^\star}$  as $\boldsymbol{W_q} = \boldsymbol{U_q \Lambda_q U_q^\ast}$, the diagonal matrix $\boldsymbol{\Lambda_q} \in \mathbb{R}^{N \times N}$ has power allocated across $K$ diagonal elements and the remaining eigenvalues for the $N-K$ beams is set to zero. Based on Theorem (\ref{theorem1}), equation (\ref{WP}b) can be rewritten as
\begin{equation}
    \text{tr}(\boldsymbol{h}_i^\ast  \boldsymbol{U_q \Lambda_q U_q^\ast} \boldsymbol{h}_i) = \pi_i
\end{equation}
where $\pi_i = \frac{e_i}{\xi T_c} \ \forall i = 1...K$. We define the row vector, $\boldsymbol{q}_i^\ast = \boldsymbol{h}_i^\ast  \boldsymbol{U_q} = \boldsymbol{h}_i^\ast  \boldsymbol{U_B}$. Then 
\begin{equation}
    \text{tr}(\boldsymbol{q_i^\ast \Lambda_q q_i}) = \pi_i
\end{equation} 
Define row vector $\boldsymbol{a_i}^\ast = \textbf{diag}(\boldsymbol{q_i q_i^\ast}) \text{ for } i = 1...K$, matrix $\boldsymbol{A} \in \mathbb{R}^{K \times K} = [\boldsymbol{a_1^\ast}...\boldsymbol{a_K^\ast}]$, and vector $\boldsymbol{b} \in \mathbb{R}^{K \times 1} = [\pi_1 ... \pi_K]$. This results in constraint (\ref{P4}c) in (P4). The ordering of $\boldsymbol{\lambda_q}$ needs to be in reverse from $\boldsymbol{\lambda_B}$, that is, in descending order, so as to minimize (\ref{Lsub}) as in (\ref{trineq}) which gives us (\ref{P4}b) in (P4).

\bibliographystyle{./IEEEtran}
\bibliography{./wptbib}
\end{document}